\documentclass[letterpaper,11pt]{article}

\usepackage{latexsym,graphics,amssymb,setspace,amsmath,psfrag}

\usepackage{fullpage}

\newcommand{\vcm}[1][1]{\vspace*{#1 cm}}
\newcommand{\hcm}[1][1]{\hspace*{#1 cm}}

\newcommand{\ceil}[1]{\left\lceil #1 \right\rceil}
\newcommand{\floor}[1]{\left\lfloor #1 \right\rfloor}
\newcommand{\f}[2]{\frac{#1}{#2}}
\newcommand{\fr}[2]{\mbox{$\frac{#1}{#2}$}}
\newcommand{\logstar}{\log^{*}}

\newcommand{\lca}{\operatorname{lca}}
\newcommand{\sens}{\operatorname{sens}}
\newcommand{\parent}{\operatorname{parent}}
\newcommand{\initialize}{\operatorname{init}}
\newcommand{\decreasekey}{\operatorname{decreasekey}}
\newcommand{\splitseq}{\operatorname{split}}
\newcommand{\findmin}{\operatorname{findmin}}
\newcommand{\sfm}{\mathsf{SF}}
\newcommand{\MST}{\mbox{\sc mst}}

\newcommand{\Komlos}{Koml\'{o}s}

\newcommand{\ignore}[1]{}

\newenvironment{proof}{\noindent {\bf Proof:}}{$\Box$}
\newtheorem{theorem}{Theorem}
\newtheorem{lemma}[theorem]{Lemma}

\title{Sensitivity Analysis of Minimum Spanning Trees in
Sub-Inverse-Ackermann Time\thanks{This work is supported by
NSF grants CCF-1217338 and CNS-1318294, and a grant from the US-Israel Binational Science Foundation.
An extended abstract appeared in the proceedings of ISAAC 2005.}}

\author{Seth Pettie\\University of Michigan}

\date{}

\begin{document}

\maketitle

\begin{abstract}
We present a deterministic algorithm for computing
the {\em sensitivity} of a minimum spanning tree (MST) or shortest path tree in 
$O(m\log\alpha(m,n))$ time, where $\alpha$ is the inverse-Ackermann
function.  This improves upon a long standing bound of $O(m\alpha(m,n))$
established by Tarjan.  Our algorithms are based on an efficient 
{\em split-findmin} data structure, which maintains a collection of 
sequences of weighted elements that may be split into smaller subsequences.
As far as we are aware, our split-findmin algorithm is the first with
superlinear but sub-inverse-Ackermann complexity.  

We also give a reduction from MST sensitivity 
to the MST problem itself.  Together with the randomized
linear time MST algorithm of Karger, Klein, and Tarjan, this gives another
randomized linear time MST sensitivity algoritm.
\end{abstract}

\section{Introduction}\label{sect:intro}

{\em Split-findmin} is a little known but key data structure in modern graph 
optimization algorithms.  It was originally designed for use in the
weighted matching and undirected all-pairs shortest path algorithms of Gabow and 
Tarjan~\cite{G85,GT91} and
has since been rediscovered as a critical component of the hierarchy-based
shortest path algorithms of Thorup \cite{Tho99}, Hagerup \cite{Hag00},
Pettie-Ramachandran \cite{PR-usssp}, and Pettie \cite{Pet02a,Pet02b}.
In this paper we apply split-findmin to the problem of performing
{\em sensitivity analysis} on minimum spanning trees (MST) and shortest path trees.
The MST sensitivity analysis problem is, given a graph $G$ and minimum spanning 
tree $T = \MST(G)$, to decide how much each individual edge weight can be perturbed 
without invalidating the identity $T=\MST(G)$.  

A thirty year old result of Tarjan \cite{Tar82} shows that
MST sensitivity analysis can be solved in $O(m\alpha(m,n))$ time, where $m$
is the number of edges, $n$ the number of vertices, and $\alpha$ the inverse-Ackermann function.  
Furthermore, he showed
that single-source shortest path sensitivity analysis can be reduced
to MST sensitivity analysis in linear time.
Tarjan's algorithm has not seen any unqualified improvements,
though Dixon et al.~\cite{DRT92} did present two MST sensitivity algorithms,
one running in {\em expected} linear time and another which is deterministic and
provably optimal, but whose complexity is only known to be bounded by $O(m\alpha(m,n))$.

In this paper we present a new MST 
sensitivity analysis algorithm running in $O(m\log\alpha(m,n))$ time.
Given the notoriously slow growth of the inverse-Ackermann function,
an improvement on the order of $\alpha/\log\alpha$ is unlikely to have
a devastating real-world impact.  Although our algorithm is simpler
and may very well be empirically 
faster than the competition, its real significance has little to do with 
practical issues, nor does it have much to do with the sensitivity
problem as such.  As one may observe in Figure~\ref{fig:incestuous},
the MST sensitivity analysis problem is related, via a tangled web of reductions, 
to many fundamental algorithmic and data structuring problems.
Among those depicted in Figure~\ref{fig:incestuous}, the two most
important unsolved problems are {\em set maxima} and
the {\em minimum spanning tree} problem.   MST sensitivity can be expressed
as a set maxima problem.  In this paper we show that MST sensitivity
is reducible to {\em both} the minimum spanning tree problem itself and 
the split-findmin data structure.  These connections suggest that it
is impossible to make progress on important optimization problems,
such as minimum spanning trees and single-source shortest paths,
without first understanding {\em why} a manifestly simpler problem like
MST sensitivity still eludes us.  In other words,
MST sensitivity should function as a test bed problem for 
experimenting with new approaches to solving its higher profile cousins.

\paragraph{Organization.}
In Section \ref{sect:theproblems} we define all the MST related problems
and data structures mentioned in Figure~\ref{fig:incestuous}.
In Section \ref{sect:sensitivity} we define the split-findmin data structure
and give new algorithms for MST sensitivity analysis 
and single-source shortest path sensitivity analysis.  
Section \ref{sect:reduction} gives a complexity-preserving
reduction from MST sensitivity analysis to the MST problem.
In Section \ref{sect:sf} we present a faster split-findmin data structure.

\begin{figure}[th]
\psfrag{Minimum Spanning Tree}{\huge Minimum Spanning Tree}
\psfrag{MST-Ver-Nontree}{\Large MST-Verification-Nontree}
\psfrag{MST-Ver-Tree}{\Large MST-Verification-Tree}
\psfrag{MST-Sens-Nontree}{\Large MST-Sensitivity-Nontree}
\psfrag{MST-Sens-Tree}{\Large MST-Sensitivity-Tree}
\psfrag{Set Maxima}{\Large Set Maxima}
\psfrag{Soft Heap Verification}{\Large Soft Heap Verification}
\psfrag{Online MST-Verification}{\Large Online MST Verification}
\psfrag{Interval-Max Data Struct}{\Large Interval-Max Data Struct.}
\psfrag{Least Common Ancestor Data Struct}{\Large LCA Data Struct.}
\psfrag{SSSP-Sensitivity}{\Large SSSP Sensitivity}
\psfrag{Split-Findmin Data Struct}{\Large Split-Findmin Data Struct.}
\psfrag{Undirected single-source shortest paths}{Undirected single-source shortest paths}
\psfrag{(Un)Directed all-pairs shortest paths}{(Un)Directed all-pairs shortest paths}
\psfrag{Weighted Matching}{Weighted matching}
\psfrag{Randomized: O(m) w.h.p.}{Randomized: O(m) w.h.p.}
\psfrag{Deterministic: optimal but unknown}{Deterministic: optimal but unknown}
\psfrag{straightline}{straightline}
\psfrag{programs}{programs}
\psfrag{randomized}{randomized}
\psfrag{+ O(nlogn)}{$+ O(n\log n)$}
\psfrag{O(malpha)}{$O(m\alpha(m,n))$}
\psfrag{Theta(m)}{$\Theta(m)$}
\psfrag{Theta(malpha)}{$\Theta(m\alpha(m,n))$}
\psfrag{O(mlogalpha)}{$O(m\log\alpha(m,n))$}
\centerline{\scalebox{.75}{\includegraphics{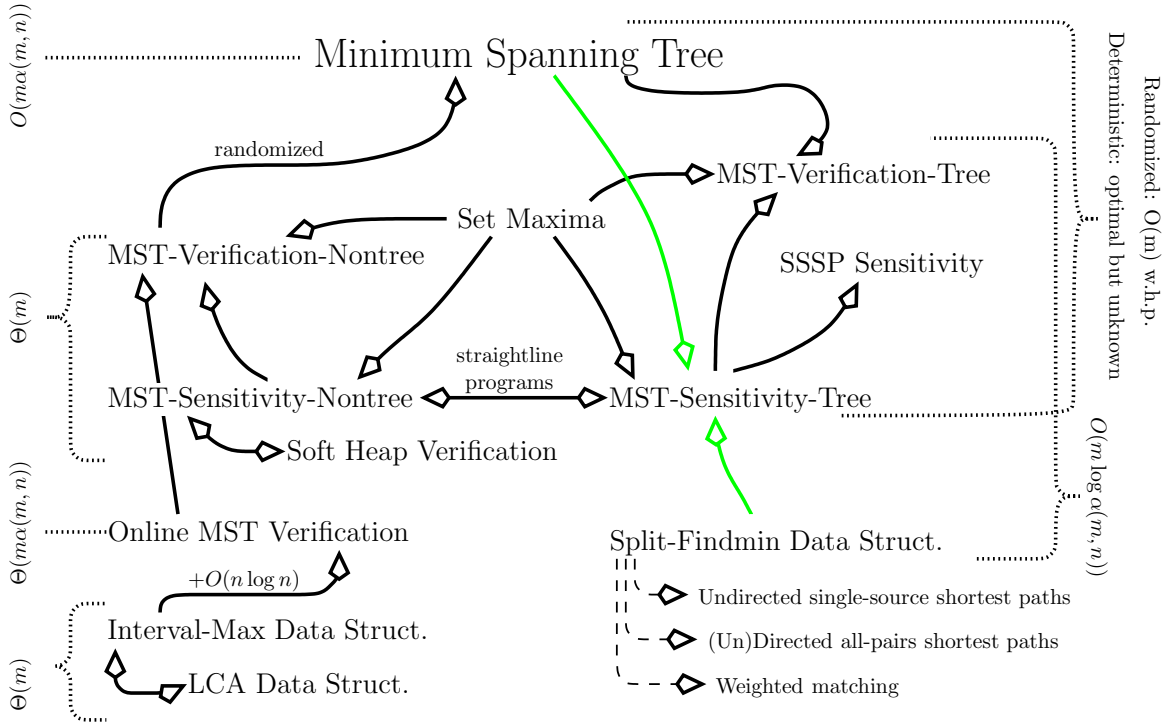}}}
\caption{\label{fig:incestuous}
The incestuous family of minimum spanning tree related
problems.  Here $\mathcal{A}\longrightarrow \mathcal{B}$ means that
problem $\mathcal{B}$ can be solved by an algorithm that uses standard
linear-time 
routines (graph contraction, least common ancestor computations, etc.)
plus calls to an algorithm for problem $\mathcal{A}$.
The dashed arrows emanating from {\em split-findmin} are not reductions;
they are only meant
to illustrate other applications of the data structure.
Some arrows are labeled with properties of the reduction.  
For instance, the reduction \cite{KKT95}
from the minimum spanning tree problem to 
MST verification of non-tree edges is randomized
and the equivalence between MST sensitivity analysis
for tree and non-tree edges only holds for straight-line programs.
The reduction from Online MST Verification to Online Interval-Max
incurs an $O(n\log n)$ term.  
}
\end{figure}

\subsection{The Problems}\label{sect:theproblems}

\begin{description}
\item[Minimum Spanning Tree.]  Given a connected undirected graph
$G=(V,E,w)$, find the spanning tree $T\subseteq E$ minimizing
$w(T) = \sum_{e\in T} w(e)$.  (For simplicity, assume throughout
that edge weights are distinct.)
In finding and verifying MSTs
there are two useful properties to keep in mind.  
{\em The Cut Property:}
the lightest edge crossing the cut $(V', V\backslash V')$ is
in $\MST(G)$, for any $V'\subset V$.  
{\em The Cycle Property:} the heaviest edge on any cycle is not in $\MST(G)$.
The best bound on the deterministic complexity of this problem is
$O(m\alpha(m,n))$, due to Chazelle \cite{Chaz00a}.  
Karger et al.~\cite{KKT95} presented a randomized MST algorithm
running in expected linear time (see also \cite{PR-randmst08})
and Pettie and Ramachandran \cite{PR02c} 
gave a deterministic MST algorithm whose 
running time is optimal and equal to the decision tree complexity of the MST problem,
i.e., somewhere between $\Omega(m)$ and $O(m\alpha(m,n))$.
See Graham and Hell \cite{GH85} for a survey on the early history of
the MST problem and Mares~\cite{Mares08} for a more recent survey.

\item[MST Verification.] 
We are given a graph $G=(V,E,w)$ and a (not necessarily minimum) spanning tree $T\subset E$.
For $e\not\in T$ let $C(e)\cup\{e\}$ be the unique cycle in $T\cup\{e\}$
and for $e\in T$ let $C^{-1}(e) = \{f\not\in T \::\: e\in C(f)\}$.
That is, $C(e)$ is the path in $T$ connecting the endpoints of $e$
and $C^{-1}(e)$ is the set of non-tree edges crossing the cut 
defined by $T\backslash\{e\}$.
In the non-tree-edge half of the problem we decide for each
$e\not\in T$ whether $e\in \MST(T\cup\{e\})$, which, by the cycle property, 
is tantamount
to deciding whether $w(e) < \max_{f\in C(e)}\{w(f)\}$.
The tree-edge version of the problem is dual: 
for each $e\in T$ we decide whether
$w(T\backslash\{e\}\cup\{f\}) < w(T)$ for some $f\in C^{-1}(e)$.

\item[MST/SSSP Sensitivity Analysis.]
We are given a weighted graph $G$ and tree $T = \MST(G)$.  The
sensitivity analysis problem is to decide how much each individual
edge weight can be altered without invalidating the identity $T=\MST(G)$.
By the cut and cycle properties it follows that we must compute
for each edge $e$:
\[
\sens(e) = \left\{
\begin{array}{l@{\hcm[.5]}l}
\max_{f\in C(e)}\{w(f)\}           & \mbox{ for $e\not\in T$}\\
&\\
\min_{f\in C^{-1}(e)}\{w(f)\}      & \mbox{ for $e\in T$}
\end{array}
\right.
\]
\noindent where $\min\emptyset = \infty$.
One can see that a non-tree edge $e$ can be increased in weight
arbitrarily or reduced by less than $w(e)-\sens(e)$.
Similarly, if $e$ is a tree-edge it can be reduced arbitrarily 
or increased by less than $\sens(e) - w(e)$.

\Komlos~\cite{Kom85} demonstrated that verification and sensitivity
analysis of {\em non-tree} edges requires a linear number of comparisons.
Linear {\em time} implementations of \Komlos's algorithm were provided
by Dixon et al.~\cite{DRT92}, King~\cite{K97}, Buchsbaum et al.~\cite{B+98}, and Hagerup~\cite{Hagerup09}.
In earlier work Tarjan \cite{Tar79b} gave a
verification/sensitivity analysis algorithm for non-tree edges that runs in time $O(m\alpha(m,n))$
and showed, furthermore, that it could be transformed \cite{Tar82} into
a verification/sensitivity analysis algorithm for tree edges with identical complexity.
(This transformation works only with straightline/oblivious algorithms and cannot
be applied to \Komlos's algorithm.)  Tarjan \cite{Tar82}
also gave a linear time reduction from
single-source shortest path sensitivity analysis to MST sensitivity analysis.

\item[Online MST Verification.]
Given a weighted tree $T$ we must preprocess it in some way so as 
to answer online queries of the form: for $e\not\in T$,
is $e\in \MST(T\cup\{e\})$?   The query edges 
$e$ are not known in advance.  
Pettie \cite{PetInvAck} proved that any data structure answering
$m$ queries must take $\Omega(m\alpha(m,n))$ time, where $n$ is the 
size of the tree.  This bound is tight \cite{Chaz87,AS87}.

\item[Soft Heap Verification.]
Pettie and Ramachandran \cite{PR-randmst08} consider a {\em generic}
soft heap (in contrast to Chazelle's concrete data structure \cite{Chaz00b})
to be any data structure that supports the priority queue operations
insert, meld, delete, and findmin, without the obligation that findmin queries
be answered correctly.  Any element that {\em bears witness} to the fact
that a findmin query was answered incorrectly is by definition {\em corrupted}.
The soft heap verification problem is, given a transcript of priority queue
operations, including their arguments and outputs, to identify the corrupted elements.
It was observed in \cite{PR-randmst08} that there are mutual reductions
between soft heap verification and MST verification (non-tree edges), 
and consequently, that soft heap verification can be solved in linear time.

\item[Online Interval-Max.]
The problem is to preprocess a sequence $(e_1,\ldots,e_n)$
of elements from a total
order such that for any two indices $\ell < r$, 
$\max_{\ell \le i \le r} \{e_i\}$ can be reported quickly.
It is known \cite{GBT84,Kom85,BF-C00} 
that answering interval-max or -min 
queries is exactly the problem of answering least common ancestor (LCA)
queries and that with linear preprocessing both
types of queries can be handled in constant time.
Katriel et al.~\cite{Katriel+03} also proved
that with an additional $O(n\log n)$ time preprocessing, 
the Online MST Verification problem can be reduced to Online Interval-Max.
This reduction was used in their minimum spanning tree algorithm.

\item[Set Maxima.]
The input is a set system (or hypergraph) $(\chi,\mathcal{S})$ where
$\chi$ is a set of $n$ weighted elements and $\mathcal{S} = \{S_1,\ldots,S_m\}$
is a collection of $m$ subsets of $\chi$.  The problem is to compute
$\{\max S_1, \ldots, \max S_m\}$ by comparing elements of $\chi$.
Goddard et al.~\cite{GKKS93} gave a {\em randomized} algorithm for set maxima
that performs $O(\min\{n\log(2\ceil{m/n}),\, n\log n\})$ comparisons, which is
optimal.  Although the dependence on randomness can be reduced \cite{PR02b},
no one has yet to produce a non-trivial deterministic algorithm.  The
current bound of $\min\{n\log n, \sum_{i=1}^m |S_i|, n + m2^m\}$ comes
from applying one of three trivial algorithms.  A natural special case of
set maxima is {\em local sorting}, where $(\chi,\mathcal{S})$ is a graph, that is, $|S_i|=2$ for all $i$.
All instances of MST verification and sensitivity analysis are reducible to 
set maxima.  In the non-tree-edge version of these problems $n$ and $m$ refer 
to the number of vertices and edges, respectively, and in their tree-edge versions
the roles of $n$ and $m$ are reversed.

\item[Split-Findmin.]
The problem is to maintain a set of disjoint sequences of weighted elements
such that the minimum weight element in each sequence is known at all times.  
(A precise definition appears in Section~\ref{sect:sensitivity}.)
The data structure can be updated in two ways: we can decrease
the weight of any element and can split any sequence into two contiguous 
subsequences.
Split-findmin could be regarded as a weighted version of 
split-find~\cite{LaP96}, which is itself a time-reversed version
of union-find~\cite{Tar75}.  
On a pointer machine union-find, split-find, and split-findmin
all have $\Omega(n+ m\alpha(m,n))$ lower bounds \cite{Tar79,LaP96}, 
where $m$ is the number of operations and $n$ the size of the structure.
The same lower bound applies to union-find \cite{FredmanS89} in the
cell-probe and RAM models.  Split-find, on the other hand, admits
a trivial linear time algorithm in the RAM model; 
see Gabow and Tarjan for the technique \cite{GT85}.
The results of this paper establish that the {\em comparison} complexity
of split-findmin is $O(n + m\log\alpha(m,n))$ and that on a RAM
there is a data structure with the same running time.
\end{description}

\paragraph{Decision Tree vs.~Semigroup Complexity.}
Many of the problems we discussed can be described
in terms of {\rm range searching} over the semigroups 
$(\mathbb{R},\max)$ and $(\mathbb{R},\min)$, where
the ranges correspond to paths.  In interval-max
and split-findmin the underlying space is 1-dimensional 
and in the non-tree-edge versions of MST Verification/Sensitivity analysis
it is a tree.  Under the assumption of an {\em arbitrary}
associative semigroup the complexities of all these problems
changes slightly.
Chazelle and Rosenberg \cite{CR91} proved that for some instances
of offline interval-sum, with $n$ elements and $m$ queries,
any algorithm must apply the semigroup's sum operator 
$\Omega(m\alpha(m,n))$ times.  The same lower bound obviously
applies to sums over tree paths.  Gabow's split-findmin \cite{G85} structure
actually solves the problem for any commutative group in
$O((n + m)\alpha(m,n))$ time.

\section{Sensitivity Analysis \& Split-Findmin}\label{sect:sensitivity}

The Split-Findmin structure maintains a set of sequences
of weighted elements.  It supports the following operations:

\begin{description}
\item[$\initialize(e_1,e_2,\ldots,e_n):$] Initialize the sequence
set $\mathcal{S} \leftarrow \{(e_1,e_2,\ldots,e_n)\}$ with 
$\kappa(e_i) \leftarrow \infty$ for all $i$.  $S(e_i)$ denotes the
unique sequence in $\mathcal{S}$ containing $e_i$.

\item[$\splitseq(e_i):$]  Let $S(e_i) = 
(e_j,\ldots,e_{i-1},e_i,\ldots,e_k)$.
Set $\mathcal{S} \leftarrow \mathcal{S} \backslash S(e_i) \cup 
\{(e_j,\ldots,e_{i-1}), (e_i,\ldots,e_k)\}$.

\item[$\findmin(e):$]  Return $\min_{f\in S(e)} \{\kappa(f)\}$.

\item[$\decreasekey(e,w):$]  Set $\kappa(e) \leftarrow \min\{\kappa(e),w\}$.
\end{description}

In Section \ref{sect:sf} we give a data structure that maintains
the minimum element in each sequence at all times.  Decreasekeys
are executed in $O(\log\alpha(m,n))$ worst-case time and
splits take $O(m/n)$ amortized time.

\subsection{Sensitivity Analysis in Sub-Inverse-Ackermann Time}

\Komlos's algorithm \cite{Kom85} and its efficient implementations
\cite{DRT92,K97,B+98} compute the sensitivity of {\em non-tree} edges in linear time.
In this subsection we calculate the sensitivity of tree edges.
We create a split-findmin structure where the initial sequence
consists of a list of the vertices in some preorder, with respect to an
arbitrary root vertex.  In general the sequences will correspond
to single vertices or subtrees of the MST.  We maintain the
invariant (through appropriate decreasekey operations) that
$\kappa(v)$ corresponds to the minimum weight edge incident to 
$v$ crossing the cut $(S(v), V\backslash S(v))$.  
If $r$ is the root of the subtree corresponding to $S(v)$
then the sensitivity of the edge $(r,\parent(r))$ can be
calculated directly from the minimum among $S(r)$.

\begin{description}
\item[Step 1.] Root the spanning tree at an arbitrary vertex; the
ancestor relation is with respect to this orientation.
For each non-tree edge $(u,v)$, unless $v$ is an ancestor of $u$
or the reverse, replace $(u,v)$ with
$(u,\lca(u,v))$ and $(v,\lca(u,v))$, 
where the new edges inherit the weight of the old.
If we have introduced multiple edges between the same endpoints
we discard all but the lightest.

\item[Step 2.]
$\initialize(u_1,\ldots,u_n)$, 
where $u_i$ is the vertex with pre-order number $i$.\footnote{Recall
that the {\em pre-order} is the order in which vertices are first
visited in some depth-first search of the tree.}  Note that for
any subtree, the pre-order numbers of vertices in that subtree form
an unbroken interval.

\item[Step 3.]\ \\
\vcm[-.7]
\begin{tabbing}
\hcm[1]\=\hcm\=\hcm\=\kill
3.1\> For $i$ from $1 $ to $ n$\\
3.2\>\>If $i>1$, $\sens(u_i,\parent(u_i)) \leftarrow \findmin(u_i)$\\
3.3\>\>Let $u_{c_1},\ldots,u_{c_\ell}$ be the children of $u_i$\\
3.4\>\>For $j$ from $1 $ to $\ell$,\\
3.5\>\>\>$\splitseq(u_{c_j})$\\
3.6\>\>For all non-tree edges $(u_k,u_i)$ where $k>i$\\
3.7\>\>\>$\decreasekey(u_k, w(u_k,u_i))$
\end{tabbing}
\end{description}

The following lemma is used to prove that correct $\sens$-values
are assigned in step 3.2.

\begin{lemma}\label{lem:sens-invariant}
Let $(u_j,\ldots,u_i,\ldots)$ be the sequence in the split-findmin
structure containing $u_i$, after an arbitrary number of iterations
of steps 3.1--3.7.  Then this sequence contains exactly
those vertices in the subtree rooted at $u_j$ and:
\[
\kappa(u_i) = \min\{w(u_i,u_k) \;:\;  k < j \mbox{ and } (u_i,u_k) \in E\}
\]
where $\min\emptyset = \infty$.  Furthermore, just before the $i$th
iteration $i=j$.
\end{lemma}

\begin{proof}
By induction on the ancestry of the tree.  
The lemma clearly holds for $i=1$, where $u_1$ is the root of the tree.
For $i>1$ the sequence containing $i$ is, by the induction hypothesis,
$(u_i,u_{c_1},\ldots,u_{c_2},\ldots,u_{c_\ell},\ldots)$.
We only need to show that the combination of the $\splitseq$s in Step 3.5
and the $\decreasekey$s in Step 3.7 ensure that the induction
hypothesis holds for iterations 
$u_{c_1},u_{c_2},u_{c_3},\ldots,u_{c_\ell}$ as well.
After performing 
$\splitseq(u_{c_1}),\ldots,\splitseq(u_{c_\ell})$
the sequences containing $u_{c_1},u_{c_2},\ldots,u_{c_\ell}$ 
clearly correspond to their respective subtrees.
Let $u_{c_j}$ be any child of $u_i$ and $u_k$ be any
vertex in the subtree of $u_{c_j}$.  
Before Step 3.7 we know, by the induction hypothesis, that:
\[
\kappa(u_k) = \min\{w(u_k,u_\nu)\;:\; \nu < i \mbox{ and } (u_k,u_\nu) \in E\}
\]
\noindent To finish the induction we must show that after Step 3.7,
$\kappa(u_k)$ is correct with respect to its new sequence beginning with $u_{c_j}$.
That is, we must consider all edges $(u_k,u_\nu)$ with $\nu < c_j$ rather than $\nu < i$.  
Since the graph is simple and all edges connect nodes to their ancestors there
can be only one edge that might affect $\kappa(u_k)$, namely $(u_k,u_i)$.
After performing $\decreasekey(u_k, w(u_k,u_i))$ we have restored the invariant
with respect to $u_k$.  Since the $i$th iteration of Step 3.1
only performs $\splitseq$s and $\decreasekey$s on elements in 
the subtree of $u_i$, all iterations in the interval
$i+1,\ldots,c_j-1$ do not have any influence on $u_{c_j}$'s sequence.
\end{proof}

\begin{theorem}
The sensitivity of a minimum spanning tree or single-source
shortest path tree can be computed in
$O(m\log\alpha(m,n))$ time, where $m$ is the number 
of edges, $n$ the number of vertices, and $\alpha$
the inverse-Ackermann function.
\end{theorem}

\begin{proof}
{\em Correctness.}  Clearly Step 1 at most 
doubles the number of edges and does not affect the 
sensitivity of any MST edge.  In
iteration $i$, $\sens(u_i,\parent(u_i))$ is set to $\findmin(u_i)$,
which is, according to Lemma~\ref{lem:sens-invariant}:
\[
\min_{\mbox{$v$ desc.~of $u_i$}} \kappa(v)
= \min_{\mbox{$v$ desc.~of $u_i$}}
\{w(v,u_k) \;:\;  k < i  \mbox{ and } (v,u_k) \in E\}
\]
\noindent which is precisely the minimum weight of any edge
whose fundamental cycle includes $(u_i,\parent(u_i))$.\\
{\em Running time.} Step 1 requires that we compute 
the least common ancestor for every pair $(u,v)\in E$.
This takes linear time \cite{HT84,BF-C00,B+98}.  
Step 2 computes a pre-order numbering in $O(n)$ time.
After Step 1 the number of non-tree edges is at most $2(m-n+1)$.
In Step 3 each non-tree edge induces one $\decreasekey$ and each 
tree vertex induces one $\findmin$ and one $\splitseq$.
By Theorem~\ref{thm:sf} the total cost of all split-findmin
operations is $O(m\log\alpha(m,n))$.
\end{proof}

\subsection{Sensitivity Analysis via Minimum Spanning Trees}\label{sect:reduction}

In this section we give a reduction from MST sensitivity analysis to the MST problem itself.
By plugging in the
Pettie-Ramachandran algorithm \cite{PR02c} this implies that the
algorithmic \& decision-tree complexities of MST sensitivity are {\em no more}
than their counterparts for the MST problem.  
By plugging in the randomized
MST algorithm of Karger et al.~\cite{KKT95} we obtain
an alternative to the randomized MST sensitivity algorithm of Dixon et al.~\cite{DRT92}.

Our reduction proceeds from a simple observation.  
Let $e\in T = \MST(G)$ be some MST edge and $V_0$ and $V_1$ 
be the connected vertex sets in $T\backslash\{e\}$.
By definition $\sens(e)$ is the weight of the minimum weight edge
crossing the cut $(V_0,V_1)$, which, by the
{\em cut property} of MSTs, must be
included in $\MST(G\backslash T)$.  To compute the sensitivity of
edges in $T$ we alternate between {\em sparsifying}
and {\em condensing} steps.

\paragraph{Sparsifying.} To sparsify we simply compute $\MST(G\backslash T)$ and solve
the MST sensitivity analysis problem recursively on $T\cup \MST(G\backslash T)$.
This yields an instance with $n$ vertices and at most $2n-2$
edges.

\paragraph{Condensing.} 
In the condensing step we reduce the number of vertices,
and thereby increase the effective density of the graph. 
Let $P = (v_1,v_2,\ldots,v_{k-1},v_k)$ be a path in $T$ satisfying the following criteria.
\begin{enumerate}
\item[(i)] $v_2,\ldots,v_{k-1}$ have degree 2 in $T$,
\item[(ii)] $v_1$ has degree one (it is a leaf) or degree at least three in $T$.
\item[(iii)] $v_k$ has degree at least three in $T$ (unless $T$ is itself a single path, in which case $P=T$ and $v_k$ is the other leaf.)
\end{enumerate}
The condensing step considers all such paths simultaneously.  We focus on one such $P$.

Call the vertices with degree one or two in $P$ {\em interior}, that is, $\{v_2,\ldots,v_{k-1}\}$ are interior
and $v_1$ is as well if it is a leaf.
All non-tree edges can be classified as {\em internal}, if both endpoints are in $P$, {\em straddling},
if one endpoint is interior to $P$ and the other outside of $P$,
and {\em external}, 
if one endpoint is not in $P$ and the other is not interior to $P$.
We eliminate each straddling edge $(u,v_i)$ by replacing it with two edges.
Without loss of generality suppose $u$ is closer to $v_k$ than $v_1$.  Replace
$(u,v_i)$ with $\{(u,v_k), (v_k,v_i)\}$, where the new edges have the same weight as $(u,v_i)$.
This transformation clearly does not affect the sensitivity of tree edges.  Note that if $v_1$ is a leaf,
it is only incident to internal non-tree edges.

We compute the sensitivity of each edge in $P$ by solving two subproblems, one for internal edges and the other on external edges.
Calculating the sensitivities with respect to internal edges can be done in linear time.  This is an easy exercise.
We form an instance of MST sensitivity analysis that consists of all external edges
and a contracted tree $T'$, obtained as follows.  
If $v_1$ is a leaf we discard $v_1,\ldots,v_{k-1}$.  If $v_1$ and $v_k$ both have degree at least three in $T$
then we replace $P$ with a single edge $(v_1,v_k)$ and discard $\{v_2,\ldots,v_{k-1}\}$.
The only vertices that remain had degree at least three in $T$, hence $|T'| < n/2$.
After both subproblems are solved $\sens(v_i,v_{i+1})$
is the minimum of $\sens(v_i,v_{i+1})$ in the internal subproblem
and $\sens(v_1,v_k)$ in the external subproblem, assuming $(v_1,v_k)$ exists in $T'$.

Define $T(m,n)$ to be the running time of this recursive MST sensitivity analysis algorithm, 
where $m$ is the number of non-tree edges and $n$ the number of vertices.
Define $\MST^\star(m,n)$ to be decision-tree complexity of the MST problem on arbitrary 
$m$-edge $n$-vertex graphs, which is the running time of the Pettie-Ramachandran algorithm~\cite{PR02c}.
The sparsifying step reduces the number of non-tree edges to $n-1$ in $O(\MST^\star(m,n))$ time
and the condensing step reduces the number of vertices by half, in $O(m+n)$ time.
We have shown that
\[
T(m,n) \,=\, O(\MST^\star(m,n)) + O(m+n) + T(n-1,(n-1)/2).
\]

The $\MST^\star(m,n)$ function is unknown, of course,
but it is simple to prove $T(m,n) = O(\MST^\star(m,n))$ given
the inequalities $\MST^\star(m,n) = \Omega(m)$ and $\MST^\star(m',n') < \MST^\star(m,n)/2$ for $m'<m/2$ and $n'<n/2$.
See Pettie and Ramachandran \cite{PR02c} for proofs
of these and other properties of $\MST^\star(m,n)$.

\section{A Faster Split-Findmin Structure}\label{sect:sf}

In this section we present a relatively simple
split-findmin data structure that runs in $O(n + m\log\alpha(m,n))$
time, where $n$ is the length of the initial sequence and
$m$ the number of operations.  
Our structure borrows many ideas from Gabow's \cite{G85} original
split-findmin data structure, whose execution time is
$O((n+m)\alpha(m,n))$ time.

The analysis makes use of Ackermann's function
and its row and column inverses:
\[
\begin{array}{rcl@{\hcm[.5]}l}
A(1,j)      & = & 2^j                         & \mbox{for $j\ge 1$}\\
A(i,1)      & = & 2                           & \mbox{for $i>1$}\\
A(i+1,j+1)  & = & A(i+1,j)\cdot A(i,A(i+1,j)) & \mbox{for $i,j\ge 1$}
\end{array}
\]

\[
\lambda_i(n) = \min\{j \,:\, A(i,j) > n\}
\mbox{\hcm[.5]and\hcm[.5]}
\alpha(m,n) = \min\{i \,:\, A(i,\ceil{\fr{2n+ m}{n}}) > n\}
\]

The definition of split-findmin from Section \ref{sect:sensitivity}
says that all $\kappa$-values are initially set to $\infty$. 
Here we consider a different version where $\kappa$-values are given.
The asymptotic complexity of these two versions is the same, of course.
However in this section we pay particular attention to the constant
factors involved.

\begin{lemma}\label{lem:sfbasis}
There is a split-findmin structure such that
decreasekeys require $O(1)$ time and $3$ comparisons
and other operations require $O(n\log n)$ time in total
and less than $3n\log n - 2n$ comparisons.
\end{lemma}

\begin{proof}
{\em The Algorithm.} 
Each sequence is divided into a set of contiguous blocks of elements.
All block sizes are powers of two and in any given sequence the
blocks are arranged in bitonic order.  From left to right the block 
sizes are strictly increasing then strictly decreasing, where the two
largest blocks may have the same size.  We maintain that each block
keeps a pointer to its minimum weight constituent element.
Similarly, each sequence keeps a pointer to its minimum weight element.
Executing a findmin is trivial.  Each decreasekey operation 
updates the key of the given element, the min-pointer of its block, 
and the min-pointer of its sequence.  Suppose we need to execute 
a split before element $e_i$, which lies in block $b$.  Unless $e_i$ is
the first element of $b$ (an easy case)
we destroy $b$ and replace it with a set of smaller blocks.
Let $b = (e_j,\ldots,e_{i-1},e_i,\ldots,e_k)$.  We scan $(e_j,\ldots,e_{i-1})$
from left to right, dividing it into at most $\log(k-j)$ blocks of decreasing size.
Similarly, we scan $(e_i,\ldots,e_k)$ from right to left, dividing it into smaller blocks.
One can easily see that this procedure preserves the bitonic order of blocks in each sequence.
To finish we update the min-pointers in each new block and new sequence.
Since one of the new sequences inherits the minimum element
from the old sequence we need only examine the other.

{\em Analysis.} It is clear that findmin and decreasekey require zero comparisons
and at most three comparisons, respectively, and that both require $O(1)$ time.
Over the life of the data structure each element belongs to at most
$\ceil{\log(n+1)}$ different blocks.  Initializing all blocks takes $O(n\log n)$
time and $\sum_{i=0}^{\floor{\log n}}\floor{\f{n}{2^i}}(2^i-1) \le n\floor{\log n} - n + 1$
comparisons.  It follows from the bitonic order of the blocks that each sequence
is made up of at most $2\floor{\log n}$ blocks; thus, updating the min-pointers
of sequences takes at most $2n\floor{\log n} - n$ comparisons in total.
\end{proof}

Note that the algorithm proposed in Lemma \ref{lem:sfbasis} is already
optimal when the number of decreasekeys is $\Omega(n\log n)$.
Lemma~\ref{lem:sfinduction} shows that any split-findmin solver can
be systematically transformed into another with substantially cheaper
splits and incrementally more expensive decreasekeys.  This is
the same type of recursion used in \cite{G85}.

\begin{lemma}\label{lem:sfinduction}
If there is a split-findmin structure that
requires $O(i)$ time and $2i+1$ comparisons per decreasekey,
and $O(in\lambda_i(n))$ time and
$3in\lambda_i(n)$ comparisons for all other operations,
then there is also a split-findmin structure with parameters
$O(i+1),2i+3,O((i+1)n\lambda_{i+1}(n)),$ and $3(i+1)n\lambda_{i+1}(n)$.
\end{lemma}

\begin{proof}
Let $\sfm_i$ and $\sfm_{i+1}$ be the assumed and derived
data structures.  
At any moment in its execution $\sfm_{i+1}$ treats
each sequence of length $n'$ as the concatenation of at most
$2(\lambda_{i+1}(n')-1)$ {\em plateaus} and at most 2 singleton elements, 
where a level $j$ plateau is partitioned into less than
$A(i+1,j+1)/A(i+1,j) = A(i,A(i+1,j))$ 
{\em blocks} of size exactly $A(i+1,j)$.  
In each sequence the plateaus are arranged in a bitonic order,
with at most two plateaus per level.
See Figure~\ref{fig:plateau} for a depiction with 6 plateaus.

\begin{figure}[ht]
\begin{center}
\psfrag{A(i+1,1)=2}{\large $A(i+1,1)=2$}
\psfrag{< A(i,2) blocks}{\large $< A(i,2)$ blocks}
\psfrag{A(i+1,2)}{\Large $A(i+1,2)$}
\psfrag{< A(i,A(i+1,2))}{\Large $< A(i,A(i+1,2))$}
\psfrag{A(i+1,3)}{\Large $A(i+1,3)$}
\psfrag{< A(i,A(i+1,3))}{\Large $< A(i,A(i+1,3))$}
\psfrag{blocks}{\Large blocks}
\psfrag{singleton}{\large singleton}
\scalebox{.6}{\includegraphics{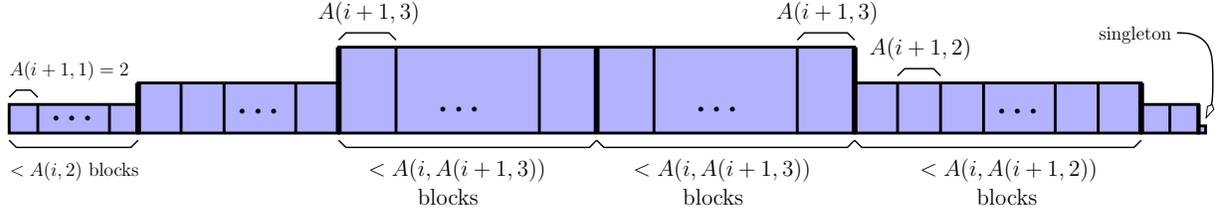}}
\end{center}
\caption{\label{fig:plateau}A sequence divided into six 
plateaus and one singleton.  The number of blocks
in a level $j$ plateau is less than $A(i+1,j+1)/A(i+1,j) = A(i,A(i+1,j))$.}
\end{figure}

At initialization $\sfm_{i+1}$ scans the whole sequence, partitioning
it into at most $\lambda_{i+1}(n)-1$ plateaus and at most one singleton.
Each plateau is managed by $\sfm_i$ as a separate instance of split-findmin,
where elements of $\sfm_i$ correspond to plateau blocks and the key of an element
is the minimum among the keys of its corresponding block.  We associate with each
plateau a pointer to the sequence that contains it.

Every block and sequence 
keeps a pointer to its minimum element.  Answering $\findmin$ queries
clearly requires no comparisons.
To execute a $\decreasekey(e,w)$ we spend one comparison
updating $\kappa(e) \leftarrow \min\{\kappa(e),w\}$ and another updating
the sequence minimum.  If $e$ is not a singleton then it is contained
in some block $b$.  We finish by calling $\decreasekey(b,w)$, where
the $\decreasekey$ function is supplied by $\sfm_i$.
If $\sfm_i$ makes $2i+1$ comparisons then $\sfm_{i+1}$ makes
$2i+3$, as promised.

Consider a split operation that divides a level $j$
block $b$ in plateau $p$.  (The splits that
occur on the boundaries between blocks and plateaus are much simpler.)
Using the $\splitseq$ operation given by $\sfm_i$, 
we split $p$ just before and after the element corresponding to $b$.
Let $b_0$ and $b_1$ be the constituent elements of $b$ to the left and right of the
splitting point.  We partition $b_0$ and $b_1$ into blocks and
plateaus (necessarily of levels less than $j$)  just as in 
the initialization procedure.  Notice that to retain the bitonic
order of plateaus we scan $b_0$ from left to right and $b_1$
from right to left.  One of the two new sequences inherits the
minimum element from the original sequeunce.  We find the minimum
of the other sequence by taking the minimum over each of its 
plateaus---this uses
$\sfm_i$'s $\findmin$ operation---and the at most two singleton elements.

The comparisons performed in split operations can be divided into
(a) those used to find block minima, (b) those used to find 
sequence minima, and (c) those performed by $\sfm_i$.
During the execution of the data structure each element appears 
in at most $\lambda_{i+1}(n)-1$ blocks.  Thus, the number of
comparisons in 
(a) is $\sum_{j\ge 1}^{\lambda_{i+1}(n)-1} (n-n/A(i+1,j))$,
which is less than $n(\lambda_{i+1}(n) - 1.5)$ since 
$A(i+1,1)=2$ for all $i$.  For (b) the number is 
$n(2\lambda_{i+1}(n) -1)$ since in any sequence 
there are at most $2(\lambda_{i+1}(n)-1)$ plateaus and 2 singletons.
For (c), notice that every element corresponding to
a block of size $A(i+1,j)$ appears in an instance of $\sfm_i$ with
less than $A(i+1,j+1)/A(i+1,j) = A(i,A(i+1,j))$ elements.  Thus the number
contributed by (c) is:
\[
\sum_{1\le j < \lambda_{i+1}(n)} \f{3in\lambda_i(A(i,A(i+1,j)) -1)}{A(i+1,j)}
= 3in(\lambda_{i+1}(n) -1)
\]
Summing up (a)--(c), the number of comparisons performed outside
of $\decreasekey$s is less than $3(i+1)n\lambda_{i+1}(n)$.
\end{proof}

\begin{theorem}\label{thm:sf}
There is a split-findmin structure that performs
$O(m\log\alpha(m,n))$ comparisons.  On a pointer machine
it runs in $O((m+n)\alpha(m,n))$ time and on a random
access machine it runs in $O(n + m\log\alpha(m,n))$ time,
where $n$ is the length of the original sequence and
$m$ the number of decreasekeys.
\end{theorem}

\begin{proof}
In conjunction, Lemmas \ref{lem:sfbasis} and \ref{lem:sfinduction}
prove that $\sfm_\alpha$ runs in
$O(\alpha(m,n)n\lambda_{\alpha(m,n)}(n) + m\alpha(m,n))$ time,
which is $O((m+n)\alpha(m,n))$ since 
$\lambda_{\alpha(m,n)}(n) = O(1 + m/n)$.
We reduce the number of comparisons in two ways, 
then improve the running time.
Suppose we are performing a $\decreasekey$
on some element $e$.  
In every $\sfm_i$ $e$ is represented in at most
one element and sequence.  Let $E_i$ and $S_i$ be
the element and sequence containing $e$ in $\sfm_i$, i.e.,
$E_i$ and $S_i$ correspond to a set of elements that include $e$.
Let $E_0$ be the block containing $e$ in $\sfm_1$.  
If we assume for simplicity that none of
$E_\alpha,E_{\alpha-1},\ldots,E_1$ correspond to singletons,
then one can easily see that
\[
\{e\} = E_{\alpha} \subseteq E_{\alpha-1} 
\subseteq \cdots \subseteq E_1 \subseteq E_0 \subseteq S_1 
\subseteq S_2 \subseteq \cdots \subseteq S_{\alpha-1} 
\subseteq S_\alpha = S(e)
\]
\noindent and therefore that
\[
E_{\alpha} \ge \min E_{\alpha-1} \ge \cdots \ge \min E_1 \ge \min E_0 \ge \min S_1 \ge \cdots 
\ge \min S_\alpha = \min S(e).
\]
Thus a decreasekey on $e$ can only affect some {\em prefix}
of the the min-pointers in $E_\alpha,\ldots,E_1,E_0,S_1,\ldots,S_\alpha$.
Using a binary search this prefix can be determined and updated
in $O(\alpha)$ time but with only $\ceil{\log(2\alpha+2)}$ comparisons.
We have reduced the number of comparisons to 
$O(n\alpha(m,n) + m\log\alpha(m,n))$.  To get rid of the 
$n\alpha(m,n)$ term we introduce another structure $\sfm_i^*$.
Upon initialization $\sfm_i^*$ divides the full sequence into
blocks of size $i$.  At any point each sequence consists of 
a subsequence of unbroken blocks and possibly two partial blocks, one
at each end.  The unbroken blocks are handled by $\sfm_i$, where
each is treated as a single element.  $\sfm_i^*$ maintains
the keys of each block in sorted order.  Findmins are easy
to handle, as are splits, which, in total, require 
$O(i(n/i)\lambda_i(n/i))=O(n\lambda_i(n))$
comparisons and $O(in + n\lambda_i(n))$ time.  The routine for 
decreasekeys requires $O(i)$ time and $O(\log i)$ comparisons.
If element $e$ lies in block $b$ then $\decreasekey(e,w)$
calls the $\sfm_i$ routine $\decreasekey(b,w)$ then updates
the sorted order of block $b$.

The $\sfm_\alpha^*$ data structure runs on a pointer machine
in $O((n+m)\alpha(m,n))$ time.  To speed it up we use the
standard RAM technique of precomputation.  The initial sequence
is divided into blocks of width $\log\log n$.  
$\sfm_2$ handles all subsequences of unbroken blocks
in $O(n\lambda_2(n/\log\log n)/\log\log n + m) = O(m) + o(n)$ time,
where $\lambda_2(n) \le \logstar n$.  Each individual block is handled
by a precomputed version of $\sfm_\alpha^*$.  We represent the
state of $\sfm_\alpha^*$ on instances of size $\log\log n$
with $o((\log\log n)^2)$ bits, which easily fits into one machine word. 
(This is an easy exercise; see \cite{DRT92,B+98,PR-usssp}.)
In $o(n)$ time we precompute the behavior of $\sfm_\alpha^*$
in a transition table.  Each entry in the table corresponds to a state
and contains $3\log\log n$ precomputed decision trees: one for each operation
(split, findmin, or decreasekey) applied to each of $\log\log n$ locations.
In the case of findmin the decision tree is trivial; it simply returns the
location of the minimum element in the given sequence.
The leaves of the decision trees for split and decreasekey operations 
point back to the appropriate entry in the transition table.  
Thus, on $\log\log n$-sized blocks the running time of $\sfm_\alpha^*$ 
is asymptotic to its comparison complexity.  
The overall running time of the data structure is $O(n + m\log\alpha(m,n))$.
\end{proof}

The use of word-packing RAM techniques is undesirable but completely 
unavoidable.  LaPoutre \cite{LaP96} has proved that on a pointer
machine, split-find requires $\Omega(m\alpha(m,n))$ time.
One can easily reduce split-find to split-findmin.
Pettie and Ramachandran \cite[Appendix B]{PR-usssp} observed that the precomputation
technique could be taken a step further.  Rather than encode the algorithm
$\sfm_\alpha^*$ we could first perform a brute force search for
the {\em optimal} split-findmin data structure, still in $o(n)$ time, 
and encode {\em it} instead.  The overall 
running time of this algorithm is still
$O(n + m\log\alpha(m,n))$ but might be asymptotically faster.

\paragraph{Acknowledgement.} I thank Bob Tarjan for encouraging me to finally publish the complete version
of this paper.


\begin{thebibliography}{10}

\bibitem{AS87}
A.~Alon and B.~Schieber.
\newblock Optimal preprocessing for answering on-line product queries.
\newblock Technical Report TR-71/87, Institute of Computer Science, Tel Aviv
  University, 1987.

\bibitem{BF-C00}
M.~A. Bender and M.~Farach-Colton.
\newblock The {LCA} problem revisited.
\newblock In {\em Proceedings 4th Latin American Symp. on Theoretical
  Informatics (LATIN), LNCS Vol. 1776}, pages 88--94, 2000.

\bibitem{B+98}
A.~L. Buchsbaum, H.~Kaplan, A.~Rogers, and J.~R. Westbrook.
\newblock Linear-time pointer-machine algorithms for {{LCA}s}, {MST}
  verification, and dominators.
\newblock In {\em Proc.~30th {ACM} Symposium on Theory of Computing ({STOC})},
  pages 279--288, May~23--26 1998.

\bibitem{Chaz87}
B.~Chazelle.
\newblock Computing on a free tree via complexity-preserving mappings.
\newblock {\em Algorithmica}, 2(3):337--361, 1987.

\bibitem{Chaz00a}
B.~Chazelle.
\newblock A minimum spanning tree algorithm with inverse-{Ackermann} type
  complexity.
\newblock {\em J.~ACM}, 47(6):1028--1047, 2000.

\bibitem{Chaz00b}
B.~Chazelle.
\newblock The soft heap: an approximate priority queue with optimal error rate.
\newblock {\em J.~ACM}, 47(6):1012--1027, 2000.

\bibitem{CR91}
B.~Chazelle and B.~Rosenberg.
\newblock The complexity of computing partial sums off-line.
\newblock {\em Internat. J. Comput. Geom. Appl.}, 1(1):33--45, 1991.

\bibitem{DRT92}
B.~Dixon, M.~Rauch, and R.~E. Tarjan.
\newblock Verification and sensitivity analysis of minimum spanning trees in
  linear time.
\newblock {\em SIAM J.~Comput.}, 21(6):1184--1192, 1992.

\bibitem{FredmanS89}
M.~L. Fredman and M.~Saks.
\newblock The cell probe complexity of dynamic data structures.
\newblock In {\em Proc. 21st annual {ACM} Symposium on Theory of Computing},
  pages 345--354, 1989.

\bibitem{G85}
H.~N. Gabow.
\newblock A scaling algorithm for weighted matching on general graphs.
\newblock In {\em Proceedings 26th IEEE Symposium on Foundations of Computer
  Science ({FOCS})}, pages 90--100, 1985.

\bibitem{GBT84}
H.~N. Gabow, J.~L. Bentley, and R.~E. Tarjan.
\newblock Scaling and related techniques for geometry problems.
\newblock In {\em Proceedings of the 16th Annual {ACM} Symposium on Theory of
  Computing}, pages 135--143, 1984.

\bibitem{GT85}
H.~N. Gabow and R.~E. Tarjan.
\newblock A linear-time algorithm for a special case of disjoint set union.
\newblock {\em J.~Comput.~Syst.~Sci.}, 30(2):209--221, 1985.

\bibitem{GT91}
H.~N. Gabow and R.~E. Tarjan.
\newblock Faster scaling algorithms for general graph-matching problems.
\newblock {\em J.~ACM}, 38(4):815--853, 1991.

\bibitem{GKKS93}
W.~Goddard, C.~Kenyon, V.~King, and L.~Schulman.
\newblock Optimal randomized algorithms for local sorting and set-maxima.
\newblock {\em SIAM J.~Comput.}, 22(2):272--283, 1993.

\bibitem{GH85}
R.~L. Graham and P.~Hell.
\newblock On the history of the minimum spanning tree problem.
\newblock {\em Ann. Hist. Comput.}, 7(1):43--57, 1985.

\bibitem{Hag00}
T.~Hagerup.
\newblock Improved shortest paths on the word {RAM}.
\newblock In {\em Proc. 27th Int'l Colloq. on Automata, Languages, and
  Programming ({ICALP}), {LNCS} vol. 1853}, pages 61--72, 2000.

\bibitem{Hagerup09}
T.~Hagerup.
\newblock An even simpler linear-time algorithm for verifying minimum spanning
  trees.
\newblock In {\em Proceedings 35th Int'l Workshop on Graph-Theoretic Concepts
  in Computer Science (WG)}, pages 178--189, 2009.

\bibitem{HT84}
D.~Harel and R.~E. Tarjan.
\newblock Fast algorithms for finding nearest common ancestors.
\newblock {\em SIAM J.~Comput.}, 13(2):338--355, 1984.

\bibitem{KKT95}
D.~R. Karger, P.~N. Klein, and R.~E. Tarjan.
\newblock {A randomized linear-time algorithm for finding minimum spanning
  trees}.
\newblock {\em J.~ACM}, 42:321--329, 1995.

\bibitem{Katriel+03}
I.~Katriel, P.~Sanders, and J.~L. Tr{\"{a}}ff.
\newblock A practical minimum spanning tree algorithm using the cycle property.
\newblock In {\em Proc.~11th Annual European Symposium on Algorithms, {LNCS}
  Vol.~2832}, pages 679--690, 2003.

\bibitem{K97}
V.~King.
\newblock A simpler minimum spanning tree verification algorithm.
\newblock {\em Algorithmica}, 18(2):263--270, 1997.

\bibitem{Kom85}
J.~Koml{\'{o}}s.
\newblock Linear verification for spanning trees.
\newblock {\em Combinatorica}, 5(1):57--65, 1985.

\bibitem{LaP96}
H.~LaPoutr\'{e}.
\newblock Lower bounds for the union-find and the split-find problem on pointer
  machines.
\newblock {\em J.~Comput.~Syst.~Sci.}, 52:87--99, 1996.

\bibitem{Mares08}
M.~Mares.
\newblock The saga of minimum spanning trees.
\newblock {\em Computer Science Review}, 2(3):165--221, 2008.

\bibitem{Pet02a}
S.~Pettie.
\newblock A faster all-pairs shortest path algorithm for real-weighted sparse
  graphs.
\newblock In {\em Proc. 29th Int'l Colloq. on Automata, Languages, and
  Programming ({ICALP}'02), {LNCS} vol. 2380}, pages 85--97, 2002.

\bibitem{Pet02b}
S.~Pettie.
\newblock On the comparison-addition complexity of all-pairs shortest paths.
\newblock In {\em Proc. 13th Int'l Symp. on Algorithms and Computation
  ({ISAAC})}, pages 32--43, 2002.

\bibitem{PetInvAck}
S.~Pettie.
\newblock An inverse-{A}ckermann type lower bound for online minimum spanning
  tree verification.
\newblock {\em Combinatorica}, 26(2):207--230, 2006.

\bibitem{PR02b}
S.~Pettie and V.~Ramachandran.
\newblock Minimizing randomness in minimum spanning tree, parallel connectivity
  and set maxima algorithms.
\newblock In {\em Proc.~13th~{ACM}-{SIAM} Symp.~on Discrete Algorithms
  ({SODA})}, pages 713--722, 2002.

\bibitem{PR02c}
S.~Pettie and V.~Ramachandran.
\newblock An optimal minimum spanning tree algorithm.
\newblock {\em J.~ACM}, 49(1):16--34, 2002.

\bibitem{PR-usssp}
S.~Pettie and V.~Ramachandran.
\newblock A shortest path algorithm for real-weighted undirected graphs.
\newblock {\em SIAM J.~Comput.}, 34(6):1398--1431, 2005.

\bibitem{PR-randmst08}
S.~Pettie and V.~Ramachandran.
\newblock Randomized minimum spanning tree algorithms using exponentially fewer
  random bits.
\newblock {\em ACM Trans. on Algorithms}, 4(1):1--27, 2008.

\bibitem{Tar75}
R.~E. Tarjan.
\newblock Efficiency of a good but not linear set merging algorithm.
\newblock {\em J.~ACM}, 22(2):215--225, 1975.

\bibitem{Tar79b}
R.~E. Tarjan.
\newblock Applications of path compression on balanced trees.
\newblock {\em J.~ACM}, 26(4):690--715, 1979.

\bibitem{Tar79}
R.~E. Tarjan.
\newblock A class of algorithms which require nonlinear time to maintain
  disjoint sets.
\newblock {\em J.~Comput.~Syst.~Sci.}, 18(2):110--127, 1979.

\bibitem{Tar82}
R.~E. Tarjan.
\newblock Sensitivity analysis of minimum spanning trees and shortest path
  problems.
\newblock {\em Info.~Proc.~Lett.}, 14(1):30--33, 1982.
\newblock See Corrigendum, {IPL} {\bf 23}(4):219.

\bibitem{Tho99}
M.~Thorup.
\newblock Undirected single-source shortest paths with positive integer weights
  in linear time.
\newblock {\em J.~ACM}, 46(3):362--394, 1999.

\end{thebibliography}

\end{document}